\documentclass[a4paper]{article}
 
\usepackage{microtype}
\usepackage{tikz}
\usetikzlibrary{arrows,automata}
\usepackage{bm}

\usepackage{amssymb,amsfonts,amsmath,amsthm}
\usepackage{fullpage}

\title{Bisimulation Metrics for Weighted Automata}

\usepackage{authblk}
\author[1]{Borja Balle}
\author[2]{Pascale Gourdeau}
\author[3]{Prakash Panangaden}
\affil[1]{Department of Mathematics and Statistics, Lancaster University, Lancaster, U.K.\\
  \texttt{b.deballepigem@lancaster.ac.uk}}
\affil[2]{School of Computer Science, McGill University, Montreal, Quebec, Canada\\
  \texttt{pascale.gourdeau@mail.mcgill.ca}}
\affil[3]{School of Computer Science, McGill University, Montreal, Quebec, Canada\\
  \texttt{prakash@cs.mcgill.ca}}

\newcommand{\R}{\mathbb{R}}
\newcommand{\N}{\mathbb{N}}
\newcommand{\cS}{\mathcal{S}}

\newcommand{\sstar}{\Sigma^{\star}}
\newcommand{\somega}{\Sigma^{\omega}}
\newcommand{\sinfty}{\Sigma^{\infty}}

\newcommand{\wa}{\langle \Sigma, V, \alpha, \beta, \{\tau_{\sigma}\}_{\sigma \in \Sigma} \rangle}
\newcommand{\hwa}{\langle \Sigma, V, \hat{\alpha}, \hat{\beta}, \{\hat{\tau}_{\sigma}\}_{\sigma \in \Sigma} \rangle}
\newcommand{\wav}{\langle \Sigma, V, v, \beta, \{\tau_{\sigma}\}_{\sigma \in \Sigma} \rangle}
\newcommand{\waw}{\langle \Sigma, V, \alpha, w, \{\tau_{\sigma}\}_{\sigma \in \Sigma} \rangle}
\newcommand{\war}{\langle \Sigma, V^*, \beta, \alpha, \{\tau_{\sigma}^\top\}_{\sigma \in \Sigma} \rangle}
\newcommand{\wan}{\langle \Sigma, V, \alpha_i, \beta_i, \{\tau_{i,\sigma}\}_{\sigma \in \Sigma} \rangle}
\newcommand{\wai}{\langle \Sigma, V_i, \alpha_i, \beta_i, \{\tau_{i,\sigma}\}_{\sigma \in \Sigma} \rangle}
\newcommand{\waai}{\langle \Sigma, V, \alpha, \beta, \tau_{i} \rangle}
\newcommand{\waa}{\langle \Sigma, V, \alpha, \beta, \tau \rangle}
\newcommand{\wap}{\langle \Sigma, V, \alpha', \beta', \{\tau_{\sigma}'\}_{\sigma
 \in \Sigma} \rangle}

\newcommand{\umdp}{\langle \Sigma, Q, \alpha, \{\beta_\sigma\}_{\sigma \in \Sigma}, \{T_{\sigma}\}_{\sigma \in \Sigma}, \gamma \rangle}
\newcommand{\umdpb}{\langle \Sigma, Q, \alpha, \beta, \{T_{\sigma}\}_{\sigma \in \Sigma}, \gamma \rangle}
\newcommand{\waumdp}{\langle \Sigma, \R^Q, \alpha, \beta, \{\tau_{\sigma}\}_{\sigma \in \Sigma} \rangle}

\newcommand{\norm}[1]{\|#1\|}
\newcommand{\bignorm}[1]{\left\|#1\right\|}
\newcommand{\spn}{\mathop{span}}

\newtheorem{theorem}{Theorem}
\newtheorem{lemma}[theorem]{Lemma}
\newtheorem{definition}[theorem]{Definition}
\newtheorem{corollary}[theorem]{Corollary}

\theoremstyle{remark}

\begin{document}

\maketitle

\begin{abstract}
  We develop a new bisimulation (pseudo)metric for weighted finite automata
  (WFA) that generalizes Boreale's linear bisimulation relation. Our metrics
  are induced by seminorms on the state space of WFA. Our development is based on spectral
  properties of sets of linear operators. In particular, the joint spectral
  radius of the transition matrices of WFA plays a central role. We also study
  continuity properties of the bisimulation pseudometric, establish an
  undecidability result for computing the metric, and give a preliminary
  account of applications to spectral learning of weighted automata.
\end{abstract}

\section{Introduction}
Weighted finite automata (WFA) form a fundamental computational model that
subsumes probabilistic automata and various other types of quantitative
automata.  They are much used in machine learning and natural language
processing, and are certainly relevant to quantitative verification and to
the theory of control systems \cite{DrosteKuich2009}. The theory of minimization of weighted
finite automata goes back to Sch\"utzenberger~\cite{Schutzenberger61} which
implicitly exploits duality as made explicit in~\cite{Bonchi14}.
In~\cite{Balle15} we began studying \emph{approximate} minimization of
WFA by using spectral methods. The idea there was to obtain automata for a
given weighted language, smaller than the minimal possible which, of
course, means that the automaton constructed does not \emph{exactly}
recognize the given weighted language but comes ``close enough.''  

In \cite{Balle15} the notion of proximity to the desired language was
captured by an $\ell_2$ distance. However, a powerful technique for
understanding approximate behavioural equivalence is by using more general
\emph{behavioural metrics}.  In particular, with a behavioural pseudometric
we recover bisimulation as the kernel.  Such behavioural metrics for Markov
processes were proposed by Giacalone et al.~\cite{Giacalone90} and the
first successful pseudometric that has bisimulation as its kernel is due to
Desharnais et al.~\cite{Desharnais99b,Desharnais04}; see
\cite{Panangaden09} for an expository account.  The subject was greatly
developed by van Breugel and Worrell~\cite{vanBreugel01a} among others.
For WFA, a beautiful treatment of linear bisimulation relations was given
by Boreale~\cite{Boreale09}.  We were motivated to develop a metric
analogue of Boreale's linear bisimulation with the eventual goal of using
it to analyze approximate minimization.  In the present paper we develop
the general theory of bisimulation (pseudo)metrics for WFA (and for
weighted languages) deferring the application to approximate minimization
to future work.

It turns out that in the linear algebraic setting appropriate to WFA it is
a (semi)norm rather than a (pseudo)metric that is the fundamental quantity
of interest.  Indeed, as one might expect, in a vector space setting norms
and seminorms are the natural objects from which metrics and pseudometrics
can be derived. The bisimulation metric that we construct actually comes
from a bisimulation seminorm which is obtained, as usual, using the Banach
fixed-point theorem. Interestingly, we also provide a closed-form
expression for the fixed point bisimulation seminorm and use it to study
several of its properties.

Our main contributions are:
\begin{enumerate}
\item The construction of bisimulation seminorms and the associated pseudometric
  on WFA (Section~\ref{sec:metrics}). The existence of the fixed point depends on some delicate
  applications of spectral theory, specifically the joint spectral radius
  of a set of matrices.
\item We obtain metrics on the space of weighted languages from the metrics
  on WFA (Section~\ref{sec:metrics}).
\item We show two continuity properties of the metric; one using definitions due to
  Jaeger et al.~\cite{Jaeger14} and the other developed here (Section~\ref{sec:continuity}).
\item We show undecidability results for computing our metrics (Section~\ref{sec:hardness}).
\item Nevertheless, we show that one can successfully exploit these metrics for
  applications in machine learning (Section~\ref{sec:learning}).
\end{enumerate}

The metric of the present paper led naturally to some sophisticated
topological and spectral theory arguments which one would not have
anticipated from the treatment of linear bisimulation in \cite{Boreale09}.

\section{Background}

In this section we recall preliminary definitions and results that will be
used throughout the rest of the paper. We assume the reader is familiar
with norms and vector spaces; these topics are
reviewed in Appendix~\ref{sec:appback}. Here we discuss Boreale's linear bisimulation
relations for weighted automata and provide a short primer on the
joint spectral radius of a set of linear operators.

\subsection{Strings and Weighted Automata}

Given a finite alphabet $\Sigma$ we let $\sstar$ denote the set of all
finite strings with symbols in $\Sigma$ and let $\sinfty$ denote the set of
all infinite strings with symbols in $\Sigma$ and we write
$\somega = \sstar \cup \sinfty$.  The length of a string $x \in \somega$ is
denoted by $|x|$; $|x| = \infty$ whenever $x \in \sinfty$.  Given a string
$x \in \somega$ and an integer $0 \leq t \leq |x|$ we write $x_{\leq t}$ to
denote the prefix containing the first $t$ symbols from $x$, with
$x_{\leq 0} = \epsilon$.  Given an integer $t \geq 0$ we will write
$\Sigma^t$ (resp.\ $\Sigma^{\leq t}$) for the set of all strings with
length equal to (resp.\ at most) $t$.
The reverse of a finite string
$x = x_1 x_2 \cdots x_t$ is given by $\bar{x} = x_t x_{t-1} \cdots x_1$.

We only consider automata with weights in the real field $\R$.  We will mostly be concerned with properties of weighted automata that are invariant under change of basis. Accordingly, our presentation uses weighted automata whose state space is an abstract real vector space.

A weighted finite automaton (WFA) is a tuple $A = \wa$ where $\Sigma$ is a
finite alphabet, $V$ is a finite-dimensional vector space, $\alpha \in V$
is a vector representing the initial weights, $\beta \in V^*$ is a linear
form representing the final weights, and $\tau_\sigma : V \to V$ is a
linear map representing the transition indexed by $\sigma \in \Sigma$.  The
vectors in $V$ are called states of $A$.  We shall denote by
$n = \dim(A) = \dim(V)$ the dimension of $A$. The transition maps
$\tau_\sigma$ can be extended to arbitrary finite strings in the obvious
way.

A weighted automaton $A = \wa$ computes the function $f_A : \sstar \to \R$ (sometimes also referred to as the weighted language in $\R^{\sstar}$ recognized by $A$) given by $f_A(x) = \beta(\tau_x(\alpha))$.  Given a WFA $A = \wa$ and a state $v \in V$ we define the weighted automaton $A_v = \wav$ obtained from $A$ by taking $v$ as the initial state. We call $f_{A_v}$ the function realized by state $v$. Similarly, give a linear form $w \in V^*$ we define the weighted automaton $A^w = \waw$ where the final weights are replaced by $w$.
The reverse of a weighted automaton $A$ is $\bar{A} = \war$, where
$\tau_\sigma^\top : V^* \to V^*$ is the transpose of $\tau_\sigma$.  It is
easy to check that the function computed by $\bar{A}$ satisfies
$f_{\bar{A}}(x) = f_A(\bar{x})$ for all $x \in \sstar$.

\subsection{Linear Bisimulations}
\label{subsec:linbisim}

Linear bisimulations for weighted automata were introduced by Boreale in \cite{Boreale09}. Here we recall the key definition and several important facts.

\begin{definition}\label{def:bisimrel}
A \emph{linear bisimulation} for a weighted automaton $A = \wa$ on a vector space $V$ is a linear subspace $W \subseteq V$ satisfying the following two conditions:
\begin{enumerate}
\item $\beta(v) = 0$ for all $v \in W$; that is, $W \subseteq \ker(\beta)$, and
\item $W$ is invariant by each $\tau_\sigma$; that is, $\tau_\sigma(W) \subseteq W$ for all $\sigma \in \Sigma$.
\end{enumerate}
Furthermore, two states $u, v \in V$ are called \emph{$W$-bisimilar} if $u - v \in W$.
\end{definition}

In particular, the trivial subspace $W = \{0\}$ is always a linear bisimulation. The notion of $W$-bisimilarity induces an equivalence relation on $V$ which we will denote by $\sim_W$. The kernel of an equivalence relation $\sim$ on a vector space $V$ is the set of vectors in the equivalence class of the null vector: $\ker(\sim) = \{ v \in V : v \sim 0 \}$. It is immediate from the definition that for any bisimulation relation $\sim_W$ we have $\ker(\sim_W) = W$.

Given a weighted automaton $A$ we say that $u, v \in V$ are $A$-bisimilar
if there exists a bisimulation $W$ for $A$ such that $u \sim_W v$. The
corresponding equivalence relation is denoted by $\sim_A$. Boreale showed
in \cite{Boreale09} that for every WFA $A$ there exists a bisimulation $W_A$
such that $\sim_{W_A}$ exactly coincides with $\sim_A$, and the
bisimulation can be obtained as $W_A = \ker(\sim_A)$. He also showed that
$W_A$ is in fact the largest linear bisimulation for $A$ in the sense that
any other linear bisimulation $W$ for $A$ must be a subspace of
$W_A$. Accordingly, we shall refer to the relation $\sim_A$ and the subspace $W_A$ as $A$-bisimulation.

Note that the subspaces considered in Definition~\ref{def:bisimrel} are independent of the initial state $\alpha$ of $A$. In fact, $A$-bisimilarity can be understood as a relation between possible initial states for $A$.
Indeed, using the definition of $\sim_A$ it is immediate to check that for
any states $u, v \in V$ we have $u \sim_A v$ if and only if
$f_{A_u} = f_{A_v}$. This implies that in a WFA where the bisimulation $W_A$
corresponding to $\sim_A$ satisfies $W_A = \{0\}$ every state realizes a
different function. Such an automaton is called observable. A weighted
automaton is called reachable if the reverse $\bar{A}$ is observable.

A weighted automaton $A$ is minimal if for any other weighted automaton
$A'$ over the same alphabet such that $f_A = f_{A'}$ we have
$\dim(A) \leq \dim(A')$. It is also shown in \cite{Boreale09} that linear
bisimulations can be used to characterize minimality, in the sense that $A$
is minimal if and only if it is observable and reachable.

\subsection{Joint Spectral Radius}\label{sec:jsr}

The joint spectral radius of a set of linear operators is a natural generalization of the spectral radius of a single linear operator. The joint spectral radius and several equivalent notions have been thoroughly studied since the 1960's. These radiuses arise in many fundamental problems in operator theory, control theory, and computational complexity. See \cite{Jungers09} for an introduction to their properties and applications. Here we recall the basic definitions and some important facts related to quasi-extremal norms.

\begin{definition}
The \emph{joint spectral radius} of a collection $M = \{\tau_i\}_{i \in I}$ of linear maps $\tau_i : V \to V$ on a normed vector space $(V,\norm{\cdot})$ is defined as
\begin{equation*}
\rho(M) = \limsup_{t \to \infty} \left(\sup_{T \in I^t} \bignorm{\prod_{i \in T} \tau_i}\right)^{1/t} = \lim_{t \to \infty} \left(\sup_{T \in I^t} \bignorm{\prod_{i \in T} \tau_i}\right)^{1/t} \enspace.
\end{equation*}
\end{definition}

The second equality above is a generalization of Gelfand's formula for the spectral radius of a single operator due to Daubechies and Lagarias \cite{Daubechies92,Daubechies01}. An important fact about the joint spectral radius is that $\rho(M)$ is independent of the norm $\norm{\cdot}$, i.e.\ one obtains the same radius regardless of the norm given to the vector space $V$. The joint spectral radius behaves nicely with respect to direct sums, in the sense that given two sets of operators $M = \{\tau_i\}_{i \in I}$ and $M' = \{\tau_i'\}_{i \in I}$, then $\rho(\{\tau_i \oplus \tau_i'\}_{i \in I}) = \max\{\rho(M),\rho(M')\}$.

The notion of joint spectral radius can be readily extended to weighted automata. Let $A = \wa$ be a weighted automaton with states on a normed vector space $(V,\norm{\cdot})$. Then the spectral radius of $A$ is defined as $\rho(A) = \rho(M)$ where $M = \{\tau_\sigma\}_{\sigma \in \Sigma}$. In this case the definition above can be rewritten as
\begin{equation*}
\rho(A) = \lim_{t \to \infty} \left(\sup_{x \in \Sigma^t} \norm{\tau_x}\right)^{1/t} \enspace.
\end{equation*}

Now we discuss several fundamental properties of the joint spectral radius that will play a role in the rest of the paper. Like in the case of the classic spectral radius, the joint spectral radius is upper bounded by the norms of the operators in $M$: $\rho(M) \leq \sup_{i \in I} \norm{\tau_i}$. Obtaining lower bounds for $\rho(M)$ is a major problem directly related to the hardness of computing approximations to $\rho(M)$. An approach often considered in the literature is to search for extremal norms. A norm $\norm{\cdot}$ on $V$ is extremal for $M$ if the corresponding induced norm satisfies $\norm{\tau_i} \leq \rho(M)$ for all $i \in I$. This immediately implies that given an extremal norm for $M$ we have $\rho(M) = \sup_{i \in I} \norm{\tau_i}$. Conditions on $M$ guaranteeing the existence of an extremal norm have been derived by Barabanov and others; see \cite{Wirth02} and references therein. However, most of these conditions are quite technical and algorithmically hard to verify. On the other hand, if one only insists on approximate extremality, the following result, which is not constructive, due to Rota and Strang guarantees the existence of such norms for any set of matrices $M$ that is compact with respect to the topology generated by the operator norm in $V$.

\begin{theorem}[\cite{Rota60}]\label{thm:quasiextremal}
Let $M = \{\tau_i\}_{i \in I}$ be a compact set of linear maps on $V$. For any $\eta > 0$ there exists a norm $\norm{\cdot}$ on $V$ that satisfies $\norm{\tau_i(v)} \leq (\rho(M) + \eta) \norm{v}$ for every $i \in I$ and every $v \in V$.
\end{theorem}

The statement above is in fact a special case of Proposition 1 in \cite{Rota60}; a proof for finite sets $M$ can be found in \cite{blondel2005accuracy}. An important result due to Barabanov \cite{barabanov1988lyapunov} states that the function $M \mapsto \rho(M)$ defined on compact sets of operators is continuous (see also \cite{heil1995continuity}). Another result that we will need was again proved by Barbanov in \cite{barabanov1988lyapunov} and it states that if $M$ is a bounded set of linear operators and $\bar{M}$ denotes its closure then $\rho(M) = \rho(\bar{M})$. Note that if $M$ is bounded then its closure $\bar{M}$ is compact by the Heine--Borel theorem.

A special case which makes the joint spectral radius easier to work with is when the set of matrices $M$ is irreducible. A set of linear maps $M$ is called \emph{irreducible} if the only subspaces $W \subseteq V$ such that $\tau_i(W) \subseteq W$ for all $i \in I$ are $W = \{0\}$ and $W = V$. If there exists a non-trivial subspace $W \subset V$ invariant by all $\tau_i$ we say that $M$ is reducible. In fact, almost all sets of matrices are irreducible in following sense. The Hausdorff distance between two sets of linear maps $M$ and $M'$ on the same normed vector space $(V,\norm{\cdot})$ is given by
\begin{equation*}
d_H(M,M') = \max\left\{ \sup_{\tau \in M} \inf_{\tau' \in M'} \norm{\tau - \tau'}, \sup_{\tau' \in M'} \inf_{\tau \in M} \norm{\tau - \tau'} \right\} \enspace.
\end{equation*}
It is possible to show that irreducible sets of matrices are dense among compact sets of matrices with respect to the topology induced by the Haussdorff distance. Furthermore, Wirth showed in \cite{Wirth02} that the joint spectral radius is locally Lipschitz continuous around irreducible sets of matrices with respect to the Hausdorff topology (see also \cite{kozyakin2010explicit} for explicit expressions for the Lipschitz constants). This can be seen as an extension of Barabanov's continuity result providing extra information about the behaviour of the function $M \mapsto \rho(M)$.

Again, the concept of irreducibility can be readily extended to WFA. We say that the weighted automaton $A = \wa$ is irreducible if $M = \{\tau_\sigma\}_{\sigma \in \Sigma}$ is irreducible. This concept will play a role in Section~\ref{sec:learning}. The following result provides a characterization of irreducibility for weighted automata in terms of minimality. In particular, the result shows that irreducibility is a stronger condition than minimality. A proof is provided in Appendix~\ref{sec:appjsr}.

\begin{theorem}\label{thm:irredwfa}
A weighted automaton $A = \wa$ is irreducible if and only if $A_v^w$ is minimal for all $v \in V$ and $w \in V^*$ with $v \neq 0$ and $w \neq 0$.
\end{theorem}

\section{Bisimulation Seminorms and Pseudometrics for WFA}\label{sec:metrics}

In the same way that the largest bisimulation relation in many settings can be obtained as a fixed point of a certain operator on equivalence relations, a possible way to define bisimulation (pseudo)metrics is via a similar fixed-point construction. See \cite{Ferns04} for an example in the case of Markov decision processes. In this section, the fixed-point construction is used to obtain a bisimulation seminorm on states of a given WFA. Given two WFA we can build their difference automaton $A$ and compute the corresponding seminorm of the initial state of $A$. This construction yields a bisimulation pseudometric between weighted automata.

Let $A = \wa$ be a weighted automaton over the vector space $V$. Let $\cS$ denote the set of all seminorms on $V$. Given $\gamma > 0$ we define the map $F_{A,\gamma} : \cS \to \cS$ between seminorms given by
\begin{equation}
F_{A,\gamma}(s)(v) = |\beta(v)| + \gamma \max_{\sigma \in \Sigma} s(\tau_\sigma(v)) \enspace.
\end{equation}
Note that this definition is independent of the initial state $\alpha$, as is the linear bisimulation for $A$ described in Section \ref{subsec:linbisim}. In the sequel we shall write $F$ instead of $F_{A,\gamma}$ whenever $A$ and $\gamma$ are clear from the context.

To verify that $F : \cS \to \cS$ is well defined we must check that the image $F(s)$ of any seminorm $s$ is also a seminorm. Absolute homogeneity is immediate by the linearity of $\beta$ and $\tau_\sigma$ and the absolute homogeneity of $s$. For the subadditivity we have
\begin{align*}
F(s)(u + v) &= |\beta(u + v)| + \gamma \max_{\sigma \in \Sigma} s(\tau_\sigma(u + v)) \\
&= |\beta(u) + \beta(v)| + \gamma \max_{\sigma \in \Sigma} s(\tau_\sigma(u) + \tau_\sigma(v)) \\
&\leq |\beta(u)| + |\beta(v)| + \gamma \max_{\sigma \in \Sigma} \left(s(\tau_\sigma(u)) + s(\tau_\sigma(v))\right) \\
&\leq F(s)(u) + F(s)(v) \enspace,
\end{align*}
where the last inequality uses subadditivity of the maximum.

To construct bisimulation seminorms for the states of a weighted automaton $A$ we shall study the fixed points of $F_{A,\gamma}$. We start by showing that $F_{A,\gamma}$ has a unique fixed point whenever $\gamma$ is small enough.

\begin{theorem}\label{thm:uniquefp}
Let $A = \wa$. If $\gamma < 1/\rho(A)$, then $F_{A,\gamma}$ has a unique fixed point.
\end{theorem}
\begin{proof}
For simplicity, let $F = F_{A,\gamma}$. By the assumption on $\gamma$ there exists some $\delta > 0$ such that $\gamma \leq 1/(\rho(A) + \delta)$. Now take $M = \{\tau_\sigma\}_{\sigma \in \Sigma}$ and $\eta = \delta / 2$ and let $\norm{\cdot}$ be the corresponding quasi-extremal norm on $V$ obtained from Theorem~\ref{thm:quasiextremal}. Using this norm we can endow $\cS$ with the metric given by $d(s,s') = \sup_{\norm{v} \leq 1} |s(v) - s'(v)|$ to obtain a complete metric space $(\cS,d)$. To see this, note that for a fixed $v$ with $\norm{v}\leq1$ the sequence $(s_n(v))$ is Cauchy, hence convergent. Call this limit $s(v)$; it is straightforward to see that this defines a seminorm. Thus, if we show that $F$ is a contraction on $\cS$ with respect to this metric, then by Banach's fixed point theorem $F$ has a unique fixed point. To see that $F$ is indeed a contraction we start by observing that:
\begin{equation}
\label{eqn:contraction1}
d(F(s),F(s')) = \sup_{\norm{v} \leq 1} |F(s)(v) - F(s')(v)| = \gamma \sup_{\norm{v} \leq 1} \left|\max_{\sigma} s(\tau_\sigma(v)) - \max_{\sigma'} s'(\tau_{\sigma'}(v))\right| \enspace.
\end{equation}
Fix any $v \in V$ with $\norm{v} \leq 1$ and suppose without loss of generality (otherwise we exchange $s$ and $s'$) that $\max_{\sigma} s(\tau_\sigma(v)) \geq \max_{\sigma'} s'(\tau_{\sigma'}(v))$. Then, letting $\sigma_*=\arg\max_\sigma s(\tau_\sigma(v))$ and using the absolute homogeneity of $s$ and $s'$, it can be shown that:
\begin{equation}\label{eqn:contraction2}
\left|\max_{\sigma} s(\tau_\sigma(v)) - \max_{\sigma'} s'(\tau_{\sigma'}(v))\right| \leq \norm{\tau_{\sigma_*}(v)} d(s,s') \enspace.
\end{equation}
We refer the reader to Appendix~\ref{sec:appmetrics} for a full derivation. Finally, we use the definition of $\norm{\cdot}$ and the choices of $\delta$ and $\eta$ to see that
\begin{align*}
\gamma \norm{\tau_{\sigma_*}(v)} \leq \gamma (\rho(A) + \eta) \norm{v}
\leq \frac{\rho(A) + \delta/2}{\rho(A) + \delta} < 1 \enspace,
\end{align*}
from which we conclude by combining \eqref{eqn:contraction1} with \eqref{eqn:contraction2} that $d(F(s),F(s')) < d(s,s')$.
\end{proof}

We now exhibit the fixed point of $F_{A,\gamma}$ in closed form. This provides a useful formula for studying properties of the resulting seminorm.

\begin{theorem}\label{thm:closedform}
Let $A = \wa$. Suppose $\gamma < 1/\rho(A)$ and let $s_{A,\gamma} \in \cS$ be the fixed point of $F_{A,\gamma}$. Then for any $v \in V$ we have
\begin{equation}\label{eqn:closeform}
s_{A,\gamma}(v) = \sup_{x \in \sinfty} \sum_{t=0}^\infty \gamma^t |\beta(\tau_{x_{\leq t}}(v))|
= \sup_{x \in \sinfty} \sum_{t=0}^\infty \gamma^t |f_{A_v}(x_{\leq t})|
\enspace.
\end{equation}
\end{theorem}

The proof can be found in Appendix~\ref{sec:appmetrics}. The next theorem is the main result of this section. It shows that any seminorm arising as a fixed point of $F_{A,\gamma}$ captures the notion of $A$-bisimulation through its kernel for any $\gamma$. Namely, two states $u, v \in V$ are $A$-bisimilar if and only if $s_{A,\gamma}(u-v)=0$. Note that this result is independent of the choice of $\gamma$, as long as the fixed point of $F_{A,\gamma}$ is guaranteed to exist.

\begin{definition}
Let $A = \wa$ be a weighted automaton with $A$-bisimulation $\sim_A$. We say that a seminorm $s$ over $V$ is a \emph{bisimulation seminorm} for $A$ if $\ker(s) = \ker(\sim_A)$.
\end{definition}

\begin{theorem}\label{thm:kers}
Let $A = \wa$. For any $0 < \gamma < 1/\rho(A)$ the fixed point $s_{A,\gamma} \in \cS$ of $F_{A,\gamma}$ is a bisimulation seminorm for $A$.
\end{theorem}
\begin{proof}
For simplicity, let $F = F_{A,\gamma}$ and $s = s_{A,\gamma}$. 
Since $W_A = \ker(\sim_A)$ is the largest bisimulation for $A$, it suffices to show that $\ker(s)$ is a bisimulation for $A$ with $W_A \subseteq \ker(s)$. For the first property we recall that $\ker(s)$ is a linear subspace of $V$ and note that for any $v\in \ker(s)$ we have, using Theorem~\ref{thm:closedform},
\begin{align*}
0 = s(v) = |\beta(v)| + \sup_{x \in \sinfty} \sum_{t=1}^\infty \gamma^t |\beta(\tau_{x_{\leq t}}(v))| \geq |\beta(v)| \geq 0 \enspace.
\end{align*}
Therefore $\ker(s) \subseteq \ker(\beta)$. Using the fact that $\beta(v)=0$, we can also verify the invariance of $\ker(s)$ under all $\tau_\sigma$, namely $s(\tau_\sigma(v))=0$ for all $v \in \ker(s)$ and $\sigma \in \Sigma$ (the full derivation is shown in Appendix~\ref{sec:appmetrics}). Therefore $\ker(s)$ is a bisimulation for $A$.

Now let $v \in W_A$. Since $W_A$ is contained in the kernel of $\beta$ and is invariant for all $\tau_\sigma$, we see that $\beta(\tau_x(v)) = 0$ for all $x \in \sstar$. Therefore, using the expression for $s$ given in Theorem~\ref{thm:closedform} we obtain $s(v) = 0$. This concludes the proof.
\end{proof}

Because every fixed point of $F_{A,\gamma}$ is a seminorm whose kernel agrees with that of Boreale's bisimulation relation $\sim_A$, we shall call them \emph{$\gamma$-bisimulation seminorms} for $A$. Interestingly, we can now show that when $A$ is observable then every $\gamma$-bisimulation seminorm is in fact a norm.

\begin{corollary}\label{cor:norm}
Let $A = \wa$ and $ \gamma < 1/\rho(A)$. If $A$ is observable then the $\gamma$-bisimulation seminorm $s_{A,\gamma}$ is a norm.
\end{corollary}
\begin{proof}
By Theorem~\ref{thm:kers} and the observability of $A$ we have $\ker(s_{A,\gamma}) = \ker(\sim_A) = \{0\}$. Thus, $s_{A,\gamma}$ is a norm.
\end{proof}

Given an automaton $A$, and state vectors $v, w \in V$, the pseudometric between states of $A$ induced by $s_{A,\gamma}$ is $d_{A,\gamma}(v,w)=s_{A,\gamma}(v-w)$. Pseudometrics of this form will be called \emph{$\gamma$-bisimulation pseudometrics}. By Corollary~\ref{cor:norm}, if $A$ is observable then $d_{A,\gamma}$ is in fact a metric.

To conclude this section we show how to use our $\gamma$-bisimulation pseudometrics to define a pseudometric between weighted automata. In order to capture the idea of distance between two WFA let us build the automaton computing the difference between their functions. Given weighted automata $A_i = \wai$ for $i = 1,2$, we define their \emph{difference automaton} as $A = A_1 - A_2 = \wa$ where $V = V_1\oplus V_2$, $\alpha = \alpha_1\oplus(-\alpha_2)$, $\beta = \beta_1\oplus\beta_2$, and $\tau_\sigma = \tau_{1,\sigma}\oplus\tau_{2,\sigma}$ for all $\sigma \in \Sigma$. Note that $A$ satisfies $f_{A}(x)=f_{A_1}(x)-f_{A_2}(x)$ for all $x \in \sstar$ and that $\rho(A) = \max\{\rho(A_1),\rho(A_2)\}$. Then, letting $s_{A,\gamma}$ be the bisimulation seminorm for $A$ we are ready to define our bisimulation distance between weighted automata.

\begin{definition}
\label{def:pseudometricWA}
Let $A_1$ and $A_2$ be two weighted automata and let $A$ be their difference automaton. For any $\gamma < 1/\rho(A)$ we define the \emph{$\gamma$-bisimulation distance} between $A_1$ and $A_2$ as $d_\gamma(A_1,A_2)=s_{A,\gamma}(\alpha)$.
\end{definition}

By exploiting the closed form expression for $s_{A,\gamma}$ given in Theorem~\ref{thm:closedform} we can provide a closed form expression for $d_{\gamma}$.

\begin{corollary}\label{cor:closedformd}
Let $A_1$ and $A_2$ two weighted automata and $\gamma < 1/\max\{\rho(A_1),\rho(A_2)\}$. Then the $\gamma$-bisimulation distance between $A_1$ and $A_2$ is given by
\begin{equation}
d_\gamma(A_1,A_2) = \sup_{x \in \sinfty} \sum_{t = 0}^{\infty} \gamma^t \left| f_{A_1}(x_{\leq t}) - f_{A_2}(x_{\leq t}) \right| \enspace.
\end{equation}
\end{corollary}

Using the properties of our bisimulation seminorms one can immediately see that $d_\gamma$ is indeed a pseudometric between all pairs of WFA such that $\gamma < 1/\rho(A_1 - A_2)$. It is also easy to see that $d_\gamma$ captures the notion of equivalence between weighted automata, in the sense that $d_\gamma(A_1, A_2) = 0$ if and only if $f_{A_1} = f_{A_2}$. Therefore, since minimal weighted automata are unique up to a change of basis, the only way to have $d_\gamma(A_1, A_2) = 0$ when $A_1$ is minimal is to have either $A_1 = A_2$ or $A_2$ is a non-minimal WFA recognizing the same weighted language as $A_1$. In particular, this implies that $d_{\gamma}$ is a metric on the set of all minimal WFA $A$ with $\gamma < 1/\rho(A)$.

\section{Continuity Properties}\label{sec:continuity}

In this section we study several continuity properties of our bisimulation pseudometrics between weighted automata. The continuity notions we consider are adapted from those presented by Jaeger et al.\ in \cite{Jaeger14}, which are developed for labelled Markov chains. Here we extend their definitions of parameter continuity and property continuity to the case of weighted automata. Such notions can be motivated by applications of metrics between transition systems to problems in machine learning \cite{Desharnais04,Ferns11,Ferns05}; see Section~\ref{sec:learning} for a discussion on how to use our bisimulation pseudometrics in the analysis of learning algorithms.

\subsection{Parameter Continuity}\label{sec:paramcont}

Given a sequence of weighted automata $A_i$ converging to a weighted automaton $A$, parameter continuity captures the notion that, as the weights in $A_i$ converge to the weights in $A$, the behavioural distance between $A_i$ and $A$ tends to zero. To make this formal we first define convergence for a sequence of automata and then parameter continuity.

\begin{definition}
Let $(A_i)_{i \in \N}$ be a sequence of WFA $A_i = \wan$ over the same alphabet $\Sigma$ and normed vector space $(V,\norm{\cdot})$. We say that the sequence $(A_i)$ \emph{converges} to $A = \wa$ if $\lim_{i \to \infty}\norm{\alpha_i -\alpha} = 0$, $\lim_{i \to \infty}\norm{\beta_i -\beta}_* = 0$, and $\lim_{i \to \infty}\norm{\tau_{i,\sigma}-\tau_\sigma} = 0$ for all $\sigma\in\Sigma$.
\end{definition}

\begin{definition} 
A pseudometric $d$ between weighted automata is \emph{parameter continuous} if for any sequence $(A_i)_{i \in \N}$ converging to some weighted automaton $A$ we have $\lim_{i \to \infty} d(A,A_i)=0$.
\end{definition}

The main result of this section is the following theorem stating that our bisimulation pseudometric $d_\gamma$ is parameter continuous.

\begin{theorem}\label{thm:parcont}
The $\gamma$-bisimulation distance between weighted automata is parameter continuous for any sequence of weighted automata $(A_i)_{i \in \N}$ converging to a weighted automaton $A$ with $\gamma < 1/\rho(A)$.
\end{theorem}

The proof of this result is quite technical and combines the following two tools:
\begin{enumerate}
\item A technical estimate of $d_{\gamma}(A,A_i)$ in terms of the distance between the weights of $A$ and $A_i$ with respect to a certain norm (Lemma~\ref{lem:boundd}). This result also plays a prominent result in Section~\ref{sec:learning}.
\item Several topological properties of the joint spectral radius discussed in Section~\ref{sec:jsr}.
\end{enumerate}
These proofs are given in Appendix~\ref{sec:appparamcont}.

\subsection{Input Continuity}\label{sec:inputcont}

Inspired by the notion of property continuity presented in \cite{Jaeger14}, input $g$-continuity encapsulates the idea that an upper bound on the behavioural distance between two systems should entail an upper bound on the difference between their outputs on any input $x \in \sstar$.

\begin{definition}\label{def:inpcon}
Let $g : \mathbb{N} \to \mathbb{R}$ be such that $g(l) > 0$ for all $l\in\mathbb{N}$. A distance function $d$ between weighted automata is \emph{input $g$-continuous} when the following holds: if $(A_i)_{i \in \N}$ is a sequence of weighted automata such that $\lim_{i \to \infty} d(A,A_i) = 0$ for some weighted automaton $A$, then one has
\begin{equation}\label{eqn:inpcon}
\lim_{i \to \infty} \sup_{x \in \sstar} \frac{| f_A(x) - f_{A_i}(x) |}{g(|x|)} = 0 \enspace.
\end{equation}
\end{definition}

Note the special case $g(l) = 1$ is tightly related to the notion of property continuity presented in \cite{Jaeger14}. The authors of that paper consider differences between the probabilities of the same event under different labelled Markov chains, and therefore always have numbers between $0$ and $1$. However, for general weighted automata the quantity $|f_A(x) - f_{A'}(x)|$ can grow unboundedly with $|x|$. Thus, in some cases we will need to have a $g(|x|)$ growing with $|x|$ in order to guarantee that \eqref{eqn:inpcon} stays bounded. The next two results show that essentially $g(|x|) = \gamma^{-|x|}$ is the threshold between input continuity and input non-continuity in our $\gamma$-bisimulation pseudometrics.

\begin{theorem}
\label{thm:gcontinuous}
The pseudometric $d_{\gamma}$ from Definition \ref{def:pseudometricWA} is
input $g$-continuous for any $g(l)=\Omega(\gamma^{-l})$. 
\end{theorem}

Note that when $\gamma > 1$ (i.e.\ when dealing with weighted automata with $\rho(A) \leq 1$) we have $g(l) = 1 \in \Omega(\gamma^{-l})$. This shows that in the case of weighted automata $A$ where every transition operator $\tau_\sigma$ can be represented by a stochastic matrix---a fact that implies $\rho(A) = 1$---our $\gamma$-bisimulation pseudometric is property continuous with respect to the definition in \cite{Jaeger14}.

Further, if $g$ does not grow fast enough as a function of the size of
$x \in \sstar$, then our bisimulation pseudometric is not input
$g$-continuous. In particular, the proof of Theorem~\ref{thm:notgcontinuous} provides simple examples of cases where $d_{\gamma}$ is not input $g$-continuous.

\begin{theorem}
\label{thm:notgcontinuous}
Let $0 < \gamma < 1$. The pseudometric $d_{\gamma}$ from Definition \ref{def:pseudometricWA} is not input $g$-continuous for any $g(l)=c^{o(l)}$ with $c > 1$.
\end{theorem}

Proofs of these results are deferred to Appendix~\ref{sec:appinputcont}.

\section{An Undecidability Result}\label{sec:hardness}

In this section we will prove that given a weighted automaton $A = \wa$, a discount factor $\gamma < 1/\rho(A)$, and a threshold $\nu > 0$, it is undecidable to check whether $s_{A,\gamma}(\alpha) > \nu$. This implies that in general the seminorms and pseudometrics studied in the previous sections are not computable.

The proof of our undecidability result involves a reduction from an undecidable planning problem. Partially observable Markov decision processes (POMDPs) are a generalization of Markov Decision Processes (MDPs) where we have a set of observations $\Omega$ and conditional observation probabilities $\mathcal{O}$. Each state emits some observation $o\in\Omega$ with a certain probability, and so we have a belief over which state we are in after taking an action and observing $o$. An MDP is a special case of a POMDP where each state has a unique observation, and an unobservable Markov decision process (UMDP) is a special case of a POMDP where all the states emit the same observation. While planning for infinite-horizon UMDPs is undecidable \cite{Madani03}, planning for finite-horizon POMDPs is decidable. 

Formally, a UMDP is a tuple $U = \umdp$ where $\Sigma$ is a finite set of actions, $Q$ is a finite set of states, $\alpha : Q \to [0,1]$ is a probability distribution over initial states in $Q$, $\beta_\sigma : Q \to \R$ represents the rewards obtained by taking action $\sigma$ from every state in $Q$, $T_\sigma : Q \times Q \to [0,1]$ is the transition kernel between states for action $\sigma$ (i.e.\ $T_\sigma(q,q')$ is the probability of transitioning to $q'$ given that action $\sigma$ is taken in $q$), and $0 < \gamma < 1$ is a discount factor. The value $V_U(x)$ of an infinite sequence of actions $x \in \sinfty$ in $U$ is the expected discounted cumulative reward collected by executing the actions in $x$ in $U$ starting from a state drawn from $\alpha$. This can be obtained as follows:
\begin{equation}
V_U(x) = \sum_{t=1}^\infty \gamma^{t-1} \alpha^\top T_{x_{\leq t-1}} \beta_{x_t} \enspace,
\end{equation}
where $T_y = T_{y_1} \cdots T_{y_t}$ for any finite string $y = y_1 \cdots y_t$ and $T_{\epsilon} = I$. The following undecidability result was proved by Madani et al.\ in \cite{Madani03}.

\begin{theorem}[Theorem 4.4 in \cite{Madani03}]\label{thm:umdpundec}
The following problem is undecidable: given a UMDP $U$ and a threshold $\nu$ decide whether there exists a sequence of actions $x \in \sinfty$ such that $V_U(x) > \nu$.
\end{theorem}

Given a UMDP $U = \umdp$, we say that $U$ has action-independent rewards
if $\beta_\sigma = \beta$ for all $\sigma \in \Sigma$. We say that $U$ has
non-negative rewards if $\beta_\sigma(q) \geq 0$ for all $q \in Q$ and
$\sigma \in \Sigma$. A careful inspection of the proof in \cite{Madani03} reveals that in fact the reduction provided in the paper always produces as output a UMDP with non-negative action-independent rewards. Thus, we have the following corollary, which forms the basis of our reduction showing that $s_\gamma$ is not computable.

\begin{corollary}\label{cor:umdpundec}
The problem in Theorem~\ref{thm:umdpundec} remains undecidable when restricted to UMDP with non-negative action-independent rewards.
\end{corollary}

\begin{theorem}
The following problem is undecidable: given a weighted automaton $A = \wa$, a discount factor $\gamma < 1/\rho(A)$, and a threshold $\nu > 0$, decide whether $s_{A,\gamma}(\alpha) > \nu$.
\end{theorem}
\begin{proof}
Let $U = \umdpb$ be a UMDP with non-negative action-independent rewards. With each UMDP of this form we associate the weighted automaton $A = \waumdp$. Here we assume that the linear form $\beta : \R^Q \to \R$ is given by $\beta(v) = v^\top \beta$, and that the linear operators $\tau_\sigma : \R^Q \to \R^Q$ are given by $\tau_\sigma(v) = v^\top T_\sigma$.

Note that the matrices $T_\sigma$ are row-stochastic and therefore we have $\rho(A) \leq \max_{\sigma} \norm{\tau_\sigma}_\infty = 1$. Thus, the discount factor in $U$ satisfies $\gamma < 1 \leq 1 / \rho(A)$ and the bisimulation seminorm $s_{A,\gamma}$ associate with $A$ is defined. Using that $U$ has non-negative action-independent rewards we can write for any $x \in \sinfty$:
\begin{align*}
V_U(x) &= \sum_{t=1}^\infty \gamma^{t-1} \alpha^\top T_{x_{\leq t-1}} \beta 
= \sum_{t=0}^\infty \gamma^{t} \alpha^\top T_{x_{\leq t}} \beta 
= \sum_{t=0}^\infty \gamma^{t} |\alpha^\top T_{x_{\leq t}} \beta| 
= \sum_{t=0}^\infty \gamma^{t} |\beta(\tau_{x \leq t}(\alpha))| \enspace.
\end{align*}
Therefore we have the relation $s_{A,\gamma}(\alpha) = \sup_{x \in \sinfty} V_U(x)$ between the bisimulation seminorm of $A$ and the value of $U$. Since deciding whether  $V_U(x) > \nu$ for some $x \in \sinfty$ is undecidable, the theorem follows.
\end{proof}

\section{Application: Spectral Learning for WFA}\label{sec:learning}

An important problem in machine learning is that of finding a weighted automaton $\hat{A}$ approximating an unknown automaton $A$ given only access to data generated by $A$.
A variety of algorithms in different learning frameworks have been considered in the literature; see \cite{cai2015} for an introductory survey. In most learning scenarios it is impossible to exactly recover the target automaton $A$ from a finite amount of data. In that case one aims for algorithms with formal guarantees of the form ``the output $\hat{A}$ automaton gets closer to $A$ as the amount of training data grows''. To prove such a result one obviously needs a way to measure the distance between two WFA. In this section we show how our $\gamma$-bisimulation pseudometric can be used to provide formal learning guarantees for a family of learning algorithms widely referred to as spectral learning. We also briefly discuss the case for behavioural metrics in automata learning problems and compare our metric to other metrics used in the spectral learning literature.

Generally speaking, spectral learning algorithms for WFA work in two phases: the first phase uses the data obtained from the target automaton $A$ to estimate a finite sub-block of the Hankel matrix of $f_A$; the second phase computes the singular value decomposition of this Hankel matrix and uses the corresponding singular vectors to solve a set of systems of linear equations yielding the weights of the output WFA $\hat{A}$. The Hankel matrix of a function $f : \sstar \to \R$ is an infinite matrix $H_{f} \in \R^{\sstar \times \sstar}$ with entries given by $H_{f}(x,y) = f(x y)$, where $x y$ denotes the string obtained by concatenating the prefix $x$ with the suffix $y$. Spectral learning algorithms work with a finite sub-block $H \in \R^{P \times S}$ of this Hankel matrix indexed by a set of prefixes $P \subset \sstar$ and a set of suffixes $S \subset \sstar$. The pair $B = (P,S)$ is usually an input to the algorithm, in which case formal learning guarantees can be provided under the assumption that $B$ is complete for $H_{f_A}$. This assumption essentially states that the sub-block of $H_{f_A}$ indexed by $B$ contains enough information to recover a WFA equivalent to $A$, and is composed of a syntactic condition ensuring $B$ contains a set of prefixes and their extensions by any symbol in $\Sigma$, and an algebraic condition ensuring the rank of the Hankel sub-matrix indexed by $B$ has the same rank as the full Hankel matrix $H_{f_A}$.
We refer the reader to \cite{cai2015,mlj13spectral} for further details about the spectral learning algorithm and a discussion of the completeness property for $B$. In the sequel we focus on the analysis of the error in the output of the spectral learning algorithm, and show how to provide learning guarantees in terms of our distance $d_{\gamma}$.

The following lemma encapsulates the first step of the analysis of spectral learning algorithms. It shows how the error between the operators of $A$ and $\hat{A}$ depends on the error between the true and the approximated Hankel matrix as measured by the standard operator $\ell_2$-norm.

\begin{lemma}\label{lem:learn}
Let $A = \wa$ be a WFA and let $H$ be a finite sub-block of the Hankel matrix $H_{f_A}$ indexed by $B = (P,S)$. Suppose $\hat{A} = \hwa$ is the WFA returned by the spectral learning algorithm using an estimation $\hat{H}$ of $H$. Let $\norm{\cdot}$ be any norm on $V$. If $B$ is complete, then we have $\norm{\alpha - \hat{\alpha}}, \norm{\beta - \hat{\beta}}_*, \max_{\sigma \in \Sigma} \norm{\tau_\sigma - \hat{\tau}_\sigma} \leq O(\norm{H - \hat{H}}_2)$ as $\norm{H - \hat{H}}_2 \to 0$. Furthermore, the constants hidden in the big-$O$ notation only depend on the norm $\norm{\cdot}$, the Hankel sub-block indices $B = (P,S)$, and the size of the alphabet $|\Sigma|$.
\end{lemma}
\begin{proof}
Combine Lemma~9.3.5 and Lemma~6.3.2 from \cite{ballethesis}.
\end{proof}

The results from \cite{ballethesis} also provide explicit expressions for the constants hidden in the big-$O$ notation. Concentration of measure for random matrices can be used to show that as the amount of training data increases then the distance between $H$ and $\hat{H}$ converges to zero with high probability (see e.g.\ \cite{denis2016}). Thus, Lemma~\ref{lem:learn} implies that as more training data becomes available, spectral learning will output a WFA $\hat{A}$ converging to $A$.

The last step in the analysis involves showing that as the weights of $\hat{A}$ get closer to the weights of $A$, the behaviour of the two automata also gets closer. Invoking the parameter continuity of $d_\gamma$ (Theorem~\ref{thm:parcont}) one readily sees that $d_\gamma(A,\hat{A}) \to 0$ as $\norm{H - \hat{H}}_2 \to 0$. This provides a proof of consistency of spectral learning with respect to the $\gamma$-bisimulation pseudometric. However, machine learning applications often require more precise information about the convergence rate of $d_\gamma(A,\hat{A})$ in order to, for example, compute the amount of data required to achieve a certain error. The following result provides such rate of convergence in the case where the target automaton is irreducible.

\begin{theorem}\label{thm:spectral}
Let $A = \wa$ be an irreducible WFA and let $H$ be a finite sub-block of the Hankel matrix $H_{f_A}$ indexed by $B = (P,S)$. Suppose $\hat{A} = \hwa$ is the WFA returned by the spectral learning algorithm using an estimation $\hat{H}$ of $H$. Suppose $B$ is complete. Then for any $\gamma < 1/\rho(A)$ we have $d_{\gamma}(A,\hat{A}) \leq O(\norm{H - \hat{H}}_2)$ as $\norm{H - \hat{H}}_2 \to 0$. Furthermore, the hidden constants in the big-$O$ notation only depend on $A$, $\gamma$, the Hankel block indices $B = (P,S)$, and the size of the alphabet $|\Sigma|$.
\end{theorem}

The local Lipschitz continuity of $\rho$ around irreducible sets of matrices plays an important role in the proof of this result (see Appendix~\ref{sec:applearning}).
Nonetheless, the irreducibility constraint is not a stringent one since the sets of irreducible matrices are known to be dense among compact sets of matrices with respect to the Hausdorff metric.

We conclude this section by comparing Theorem~\ref{thm:spectral} with analyses of spectral learning based on other error measures.
We start by noting that all finite-sample analyses of spectral learning for WFA we are aware of in the literature provide error bounds in terms of some finite variant of the $\ell_1$ distance. In particular, the analyses in \cite{hsu2012spectral,siddiqi10} bound $\sum_{x \in \Sigma^t} |f_A(x) - f_{\hat{A}}(x)|$ for a fixed $t \geq 0$, while the analyses in \cite{baillythesis,ballethesis,glaude2016pac} extend the bounds to $\sum_{x \in \Sigma^{\leq t}} |f_A(x) - f_{\hat{A}}(x)|$ for a fixed $t \geq 0$. This approach poses several drawbacks, including:
\begin{enumerate}
\item Finite $\ell_1$-norms provide a pseudo-metric between WFA whose kernel includes pairs of non-equivalent WFA.
\item The number of samples required to achieve a certain error increase with the horizon $t$, meaning that more data is required to get the same error on longer strings, and that existing bounds become vacuous in the case $t \to \infty$.
\end{enumerate}
In contrast, our result in terms of $d_\gamma$ establishes a bound on the discrepancy between $A$ and $\hat{A}$ on strings of arbitrary length and will never assign zero distance to a pair of automata realizing different functions. Furthermore, our bisimulation metric still makes sense outside the setting of spectral learning of probabilistic automata where most of the techniques mentioned above have been developed.

\section{Conclusion}

The metric developed in this paper was very much motivated and informed by
spectral ideas. Not surprisingly it was well suited for analyzing spectral
learning algorithms for weighted automata. Two obvious directions for future work are:
\begin{enumerate}
\item Approximation algorithms for the bisimulation metric.
\item Exploring the relation to approximate minimization.
\end{enumerate}
Both of these are well underway. It seems that some recent ideas from
non-linear optimization are very useful in developing approximation
algorithms and we hope to be able to report our results soon. Exploring
the relation to approximate minimization is less further along, but the
spectral ideas at the heart of the approximate minimization
algorithm in \cite{Balle15} should be well adapted to the techniques of the present paper.

\subparagraph*{Acknowledgements.}

We would like to thank Doina Precup who was actively involved in the
approximate minimization work.  This research has been supported by a grant
from NSERC (Canada).

\bibliography{prakash}

\begin{thebibliography}{10}

\bibitem{baillythesis}
R.~Bailly.
\newblock {\em M{\'e}thodes spectrales pour l'inf{\'e}rence grammaticale
  probabiliste de langages stochastiques rationnels}.
\newblock PhD thesis, Aix-Marseille Universit{\'e}, 2011.

\bibitem{ballethesis}
Borja Balle.
\newblock {\em Learning Finite-State Machines: Algorithmic and Statistical
  Aspects}.
\newblock PhD thesis, Universitat Polit{\`e}cnica de Catalunya, 2013.

\bibitem{mlj13spectral}
Borja Balle, Xavier Carreras, Franco~M. Luque, and Ariadna Quattoni.
\newblock Spectral learning of weighted automata: A forward-backward
  perspective.
\newblock {\em Machine Learning}, 2014.

\bibitem{cai2015}
Borja Balle and Mehryar Mohri.
\newblock Learning weighted automata.
\newblock In {\em Conference on Algebraic Informatics}, 2015.

\bibitem{Balle15}
Borja Balle, Prakash Panangaden, and Doina Precup.
\newblock A canonical form for weighted automata and applications to
  approximate minimization.
\newblock In {\em Proceedings of the Thirtieth Annual ACM-IEEE Symposium on
  Logic in Computer Science}, July 2015.

\bibitem{barabanov1988lyapunov}
Nikita~E. Barabanov.
\newblock {On the Lyapunov indicator of discrete inclusions, part I, II, and
  III}.
\newblock {\em Avtomatika i Telemekhanika}, 2:40--46, 1988.

\bibitem{blondel2005accuracy}
Vincent~D. Blondel, Yurii Nesterov, and Jacques Theys.
\newblock On the accuracy of the ellipsoid norm approximation of the joint
  spectral radius.
\newblock {\em Linear Algebra and its Applications}, 394:91--107, 2005.

\bibitem{Bonchi14}
Filippo Bonchi, Marcello~M. Bonsangue, Helle~Hvid Hansen, Prakash Panangaden,
  Jan Rutten, and Alexandra Silva.
\newblock Algebra-coalgebra duality in {B}rzozowski's minimization algorithm.
\newblock {\em ACM Transactions on Computational Logic}, 2014.

\bibitem{Boreale09}
Michele Boreale.
\newblock Weighted bisimulation in linear algebraic form.
\newblock In {\em CONCUR 2009-Concurrency Theory}, pages 163--177. Springer,
  2009.

\bibitem{Daubechies92}
Ingrid Daubechies and Jeffrey~C. Lagarias.
\newblock Sets of matrices all infinite products of which converge.
\newblock {\em Linear algebra and its applications}, 161:227--263, 1992.

\bibitem{Daubechies01}
Ingrid Daubechies and Jeffrey~C. Lagarias.
\newblock Corrigendum/addendum to: Sets of matrices all infinite products of
  which converge.
\newblock {\em Linear Algebra and its Applications}, 327(1-3):69--83, 2001.

\bibitem{denis2016}
Fran\c{c}ois Denis, Mattias Gybels, and Amaury Habrard.
\newblock Dimension-free concentration bounds on hankel matrices for spectral
  learning.
\newblock {\em Journal of Machine Learning Research}, 17(31):1--32, 2016.

\bibitem{Desharnais99b}
J.~Desharnais, V.~Gupta, R.~Jagadeesan, and P.~Panangaden.
\newblock Metrics for labeled {Markov} systems.
\newblock In {\em Proceedings of CONCUR99}, number 1664 in Lecture Notes in
  Computer Science. Springer-Verlag, 1999.

\bibitem{Desharnais04}
Jos\'ee Desharnais, Vineet Gupta, Radhakrishnan Jagadeesan, and Prakash
  Panangaden.
\newblock A metric for labelled {Markov} processes.
\newblock {\em Theoretical Computer Science}, 318(3):323--354, June 2004.

\bibitem{DrosteKuich2009}
Manfred Droste, Werner Kuich, and Heiko Vogler, editors.
\newblock {\em Handbook of weighted automata}.
\newblock {EATCS} Monographs on Theoretical Computer Science. Springer, 2009.

\bibitem{Ferns04}
Norm Ferns, Prakash Panangaden, and Doina Precup.
\newblock Metrics for finite {M}arkov decision processes.
\newblock In {\em Proceedings of the 20th conference on Uncertainty in
  Artificial Intelligence}, pages 162--169. AUAI Press, 2004.

\bibitem{Ferns05}
Norm Ferns, Prakash Panangaden, and Doina Precup.
\newblock Metrics for {M}arkov decision processes with infinite state spaces.
\newblock In {\em Proceedings of the 21st Conference on Uncertainty in
  Artificial Intelligence}, pages 201--208, July 2005.

\bibitem{Ferns11}
Norm Ferns, Prakash Panangaden, and Doina Precup.
\newblock Bisimulation metrics for continuous markov decision processes.
\newblock {\em SIAM Journal on Computing}, 40(6):1662--1714, 2011.

\bibitem{Giacalone90}
A.~Giacalone, C.~Jou, and S.~Smolka.
\newblock Algebraic reasoning for probabilistic concurrent systems.
\newblock In {\em Proceedings of the Working Conference on Programming Concepts
  and Methods}, IFIP TC2, 1990.

\bibitem{glaude2016pac}
Hadrien Glaude and Olivier Pietquin.
\newblock Pac learning of probabilistic automaton based on the method of
  moments.
\newblock In {\em Proceedings of The 33rd International Conference on Machine
  Learning}, pages 820--829, 2016.

\bibitem{heil1995continuity}
Christopher Heil and Gilbert Strang.
\newblock Continuity of the joint spectral radius: application to wavelets.
\newblock In {\em Linear Algebra for Signal Processing}, pages 51--61.
  Springer, 1995.

\bibitem{hsu2012spectral}
Daniel Hsu, Sham~M Kakade, and Tong Zhang.
\newblock A spectral algorithm for learning hidden markov models.
\newblock {\em Journal of Computer and System Sciences}, 78(5), 2012.

\bibitem{Jaeger14}
Manfred Jaeger, Hua Mao, Kim~Guldstrand Larsen, and Radu Mardare.
\newblock Continuity properties of distances for {M}arkov processes.
\newblock In {\em Proceedings of QEST 2014 Quantitative Evaluation of Systems:
  11th International Conference}, pages 297--312. Springer International
  Publishing, 2014.

\bibitem{Jungers09}
Rapha{\"e}l Jungers.
\newblock {\em The joint spectral radius: theory and applications}, volume 385.
\newblock Springer Science and Business Media, 2009.

\bibitem{kozyakin2010explicit}
Victor Kozyakin.
\newblock An explicit lipschitz constant for the joint spectral radius.
\newblock {\em Linear Algebra and its Applications}, 433(1):12--18, 2010.

\bibitem{Madani03}
Omid Madani, Steve Hanks, and Anne Condon.
\newblock On the undecidability of probabilistic planning and related
  stochastic optimization problems.
\newblock {\em Artificial Intelligence}, 147(1-2):5--34, 2003.

\bibitem{Panangaden09}
Prakash Panangaden.
\newblock {\em Labelled {M}arkov Processes}.
\newblock Imperial College Press, 2009.

\bibitem{Rota60}
Gian-Carlo Rota and W.~Strang.
\newblock A note on the joint spectral radius.
\newblock {\em Indag. Math.}, 22:379–381, 1960.

\bibitem{Schutzenberger61}
Marcel~Paul Sch{\"u}tzenberger.
\newblock On the definition of a family of automata.
\newblock {\em Information and control}, 4(2):245--270, 1961.

\bibitem{siddiqi10}
S.~M. Siddiqi, B.~Boots, and G.~Gordon.
\newblock Reduced-rank hidden {M}arkov models.
\newblock In {\em AISTATS}, 2010.

\bibitem{vanBreugel01a}
Franck van Breugel and James Worrell.
\newblock Towards quantitative verification of probabilistic systems.
\newblock In {\em Proceedings of the Twenty-eighth International Colloquium on
  Automata, Languages and Programming}. Springer-Verlag, July 2001.

\bibitem{Wirth02}
Fabian Wirth.
\newblock The generalized spectral radius and extremal norms.
\newblock {\em Linear Algebra and its Applications}, 342(1-3):17--40, 2002.

\end{thebibliography}

\appendix

\section{Norms, Seminorms, and Pseudometrics}\label{sec:appback}

Seminorms (resp.\ pseudometrics) are generalizations of norms (resp.\ metrics) often used in analysis. The key difference is that seminorms (resp.\ pseudometrics) are allowed to assign zero value to non-zero vectors (resp.\ zero distance to pairs of distinct vectors). This section recalls their definitions and main properties.

Given a finite-dimensional normed real vector space $(V,\norm{\cdot})$ we let $V^*$ denote the dual vector space equipped with the dual norm $\norm{w}_* = \sup_{\norm{v} \leq 1} w(v)$ for any $w \in V^*$. The induced norm of a linear operator $\tau : V \to V$ is defined as $\norm{\tau} = \sup_{\norm{v} \leq 1} \norm{\tau(v)}$. We recall that on a finite-dimensional vector space all norms are equivalent. Namely, given two norms $\norm{\cdot}$ and $\norm{\cdot}'$ on $V$ there exists a pair of constants $0 < c \leq C$ such that $c \norm{v} \leq \norm{v}' \leq C \norm{v}$ holds for all $v \in V$. It is immediate to check that the inequalities $C^{-1} \norm{w}_* \leq \norm{w}'_* \leq c^{-1} \norm{w}_*$ hold for the corresponding dual norms.

A seminorm $s$ on a vector space $V$ is a function $s : V \to \R$ satisfying the following two conditions:
\begin{enumerate}
\item (absolute homogeneity) $s(c v) = |c| s(v)$ for all $c \in \R$ and $v \in V$, and
\item (subadditivity) $s(u + v) \leq s(u) + s(v)$ for all $u, v \in V$.
\end{enumerate}
Jointly, these two conditions imply $s(v) \geq 0$ for all $v \in V$. Furthermore, the first condition implies $s(0) = 0$, but unlike in the case of norms we do not require that $0$ is the only vector with $s(v) = 0$. The kernel of a seminorm $s$ is defined as $\ker(s) = \{ v \in V : s(v) = 0\}$. Therefore, a seminorm $s$ is a norm if and only if $\ker(s) = \{0\}$. It can be readily verified that $\ker(s)$ is always a linear subspace of $V$.

A pseudometric on a set $V$ is a function $d : V \times V \to \R$ satisfying the following conditions:
\begin{enumerate}
\item (non-negativity) $d(v,w) \geq 0$ for all $v, w \in V$,
\item (indiscernibility of identicals) $d(v,v) = 0$ for all $v \in V$,
\item (symmetry) $d(v,w) = d(w,v)$ for all $v, w \in V$, and
\item (triangle inequality) $d(v,u) \leq d(v,w) + d(w,u)$ for all $u, v, w \in V$.
\end{enumerate}
Note that the only difference between a metric and a pseudometric is that in the latter case we do not require that $d(v,w) = 0$ implies $v = w$. Therefore, a pseudometric might not be able to distinguish between every pair of points in $V$. Seminorms provide a convenient way to build pseudometrics: if $V$ is a real vector space and $s : V \to \R$ is a seminorm on $V$, then $d(v,w) = s(v - w)$ is a pseudometric on $V$. We shall say that $d$ is the pseudometric induced by $s$.

\section{Proofs from Section~\ref{sec:jsr}}\label{sec:appjsr}

The following characterizations of reachability and observability will be used in the proof.

\begin{lemma}\label{lemma:obsreach}
Given a weighted automaton $A = \wa$ the following hold:
\begin{enumerate}
\item $A$ is observable if and only if $f_{A_v} \neq 0$ for all $v \in V \setminus \{0\}$.
\item $A$ is reachable if and only if $f_{A^w} \neq 0$ for all $w \in V^* \setminus \{0\}$.
\end{enumerate}
\end{lemma}
\begin{proof}
To prove the first claim we note that if $A$ is not observable then there exist two different states $u, v \in V$ such that $f_{A_u} = f_{A_v}$. Therefore, we see that $w = u - v \neq 0$ and $A_w$ computes the function $f_{A_w} = f_{A_u} - f_{A_v} = 0$. On the other hand, if $v \in V \setminus \{0\}$ is such that $f_{A_v} = 0$, then $A_v$ and $A_0$ compute the zero function and $A$ is not observable.

The second claims follows from applying the first claim to the reverse automaton $\bar{A}$.
\end{proof}

\begin{proof}[Proof of Theorem~\ref{thm:irredwfa}]
To prove the ``only if'' part assume that the set of linear maps $M = \{\tau_\sigma\}_{\sigma \in \Sigma}$ is reducible. Then there exists a non-trivial subspace $W \subset V$ that is left invariant by all the $\tau_\sigma$. Using this subspace we can find a non-zero vector $v \in W$ and a non-zero linear form $w \in V^*$ such that $W \subseteq \ker(w)$. We claim that $A' = A_v^w$ is not minimal. Indeed, since $W$ is invariant by every $\tau_\sigma$ we have $\tau_x(v) \in W$ for all $x \in \sstar$, which implies $f_{A'}(x) = w(\tau_x(v)) = 0$ for all $x \in \sstar$. Therefore we have $f_{A'} = 0$ which is also computed by the weighted automaton $A_0^w$ with initial weights $0 \in V$, so $A'$ is not observable.

For the ``if'' part we assume that $A_v^w$ is not minimal for some $v \in V \setminus \{0\}$ and $w \in V^* \setminus \{0\}$. Since $A$ is irreducible if and only if $\bar{A}$ is irreducible, we can assume without loss of generality that $A_v^w$ is not observable. Furthermore, by Lemma~\ref{lemma:obsreach} we can further assume that (replacing $v$ by a different state if necessary) $A_v^w$ computes the zero function. Now let us take the subspace $W = \spn \{ \tau_x(v) : x \in \sstar \} \subseteq V$ and show that it is a witness for the reducibility of $M$. Note that by construction we immediately have $\tau_\sigma(W) \subseteq W$ for any $\sigma \in \Sigma$, so we only need to check that $W$ is not trivial. On the one hand we have $0 \neq v \in W$, so $\dim(W) \geq 1$. On the other hand, since $A_v^w$ computes the zero function we must have $W \subseteq \ker(w)$, which implies $\dim(W) \leq \dim(\ker(w)) = n - 1$ since $w$ is not zero.
\end{proof}

\section{Proofs from Section~\ref{sec:metrics}}\label{sec:appmetrics}

\begin{proof}[Proof of Theorem~\ref{thm:uniquefp}]
For simplicity, let $F = F_{A,\gamma}$. By the assumption on $\gamma$ there exists some $\delta > 0$ such that $\gamma \leq 1/(\rho(A) + \delta)$. Now take $M = \{\tau_\sigma\}_{\sigma \in \Sigma}$ and $\eta = \delta / 2$ and let $\norm{\cdot}$ be the corresponding quasi-extremal norm on $V$ obtained from Theorem~\ref{thm:quasiextremal}. Using this norm we can endow $\cS$ with the metric given by $d(s,s') = \sup_{\norm{v} \leq 1} |s(v) - s'(v)|$ to obtain a complete metric space $(\cS,d)$. Thus, if we show that $F$ is a contraction on $\cS$ with respect to this metric, then by Banach's fixed point theorem $F$ has a unique fixed point. To see that $F$ is indeed a contraction we start by observing that:
\begin{equation}
\label{eqn:contraction1a}
d(F(s),F(s')) = \sup_{\norm{v} \leq 1} |F(s)(v) - F(s')(v)| = \gamma \sup_{\norm{v} \leq 1} \left|\max_{\sigma} s(\tau_\sigma(v)) - \max_{\sigma'} s'(\tau_{\sigma'}(v))\right| \enspace.
\end{equation}
Fix any $v \in V$ with $\norm{v} \leq 1$ and suppose without loss of generality (otherwise we exchange $s$ and $s'$) that $\max_{\sigma} s(\tau_\sigma(v)) \geq \max_{\sigma'} s'(\tau_{\sigma'}(v))$. Then, using the absolute homogeneity of $s$ and $s'$, it can be shown that:
\begin{align}
\left|\max_{\sigma} s(\tau_\sigma(v)) - \max_{\sigma'} s'(\tau_{\sigma'}(v))\right| &=
\max_{\sigma} s(\tau_\sigma(v)) - \max_{\sigma'} s'(\tau_{\sigma'}(v)) \notag \\
&= s(\tau_{\sigma_*}(v)) - \max_{\sigma'} s'(\tau_{\sigma'}(v)) \notag \\
&\leq s(\tau_{\sigma_*}(v)) - s'(\tau_{\sigma_*}(v)) \notag \\
&= \norm{\tau_{\sigma_*}(v)} \left(s\left(\frac{\tau_{\sigma_*}(v)}{\norm{\tau_{\sigma_*}(v)}}\right) - s'\left(\frac{\tau_{\sigma_*}(v)}{\norm{\tau_{\sigma_*}(v)}}\right)\right) \notag \\
&\leq \norm{\tau_{\sigma_*}(v)} \sup_{\norm{v'} \leq 1} | s(v') - s'(v') | \notag \\
&= \norm{\tau_{\sigma_*}(v)} d(s,s') \enspace.\label{eqn:contraction2a}
\end{align}
We refer the reader to the appendix for a full derivation. Finally, we use the definition of $\norm{\cdot}$ and the choices of $\delta$ and $\eta$ to see that
\begin{align*}
\gamma \norm{\tau_{\sigma_*}(v)} \leq \gamma (\rho(A) + \eta) \norm{v}
\leq \frac{\rho(A) + \delta/2}{\rho(A) + \delta} < 1 \enspace,
\end{align*}
from which we conclude by combining \eqref{eqn:contraction1a} with \eqref{eqn:contraction2a} that $d(F(s),F(s')) < d(s,s')$.
\end{proof}

\begin{proof}[Proof of Theorem~\ref{thm:closedform}]
For simplicity, let $F=F_{A,\gamma}$ and $s=s_{A,\gamma}$. In the first place we note that $s$ clearly satisfies the seminorm axioms. However, this is not enough to guarantee that $s$ is a seminorm because the supremum over $\sinfty$ could be unbounded while the definition of seminorm requires the image by $s$ of every element in $V$ to be in $\R$. To guarantee that $s$ is a seminorm we must show that $s(v)$ is always finite. Let $\norm{\cdot}$ be the norm on $V$ constructed in the proof of Theorem~\ref{thm:uniquefp}. Then we can use H\"older's inequality and the submultiplicativity of induced norms to show that for any $v \in V$ and $x \in \sstar$ we have
\begin{equation*}
|\beta(\tau_x(v))| \leq \norm{\tau_x(v)} \norm{\beta}_* \leq (\rho(A) + \eta)^{|x|} \norm{v} \norm{\beta}_* \enspace,
\end{equation*}
where $\eta = \delta / 2$ for some $\delta > 0$ such that $\gamma \leq 1 / (\rho(A) + \delta)$. Thus, for any $v \in V$ we can bound the expression in \eqref{eqn:closeform} as
\begin{align*}
s(v)
\leq \norm{v} \norm{\beta}_* \sum_{t = 0}^\infty \gamma^t (\rho(A) + \eta)^t
\leq \norm{v} \norm{\beta}_* \sum_{t=0}^\infty \left(\frac{\rho(A) + \delta/2}{\rho(A) + \delta}\right)^t
< \infty \enspace.
\end{align*}
Now that we know that $s$ is a seminorm and $F$ has a unique fixed point in $\cS$, we only need to verify that the expression in \eqref{eqn:closeform} is a fixed point of $F$. To see that this is the case we just note the following holds for any $v\in V$:
\begin{align*}
F(s)(v) &= |\beta(v)| + \gamma \max_{\sigma \in \Sigma} |s(\tau_{\sigma}(v))| \\
&= |\beta(v)| + \gamma \max_{\sigma \in \Sigma} \left| \sup_{x \in \sinfty} \sum_{t=0}^\infty \gamma^t |\beta(\tau_{x_{\leq t}}(\tau_\sigma(v)))| \right| \\
&=  |\beta(v)| + \max_{\sigma \in \Sigma} \sup_{x \in \sinfty} \sum_{t=0}^\infty \gamma^{t+1} |\beta(\tau_{(\sigma x)_{\leq t+1}}(v))| \\
&= |\beta(v)| + \sup_{x \in \sinfty} \sum_{t=1}^\infty \gamma^{t} |\beta(\tau_{x_{\leq t}}(v))| \\
&= s(v) \enspace.
\end{align*}
Finally, note that the second equality follows from the identity $|\beta(\tau_y(v))| = f_{A_v}(y)$ for all $y \in \sstar$.
\end{proof}

\begin{proof}[Proof of Theorem ~\ref{thm:kers}]
For simplicity, let $F = F_{A,\gamma}$ and $s = s_{A,\gamma}$. 
Since $W_A = \ker(\sim_A)$ is the largest bisimulation for $A$, it suffices to show that $\ker(s)$ is a bisimulation for $A$ with $W_A \subseteq \ker(s)$. For the first property we recall that $\ker(s)$ is a linear subspace of $V$ and note that for any $v\in \ker(s)$ we have, using Theorem~\ref{thm:closedform},
\begin{align*}
0 = s(v) = |\beta(v)| + \sup_{x \in \sinfty} \sum_{t=1}^\infty \gamma^t |\beta(\tau_{x_{\leq t}}(v))| \geq |\beta(v)| \geq 0 \enspace.
\end{align*}
Therefore $\ker(s) \subseteq \ker(\beta)$. To verify the invariance of $\ker(s)$ under all $\tau_\sigma$ let $v \in \ker(s)$ and note that using $\beta(v) = 0$ we can write
\begin{align*}
0 \leq s(\tau_\sigma(v)) &=
\sup_{x \in \sinfty} \sum_{t=0}^\infty \gamma^t |\beta(\tau_{x_{\leq t}}(\tau_\sigma(v)))| \\
&= \sup_{x \in \sinfty} \sum_{t=0}^\infty \gamma^t |\beta(\tau_{(\sigma x)_{\leq t+1}}(v))| \\
&= \frac{1}{\gamma} \sup_{x \in \sinfty} \sum_{t=0}^\infty \gamma^{t+1} |\beta(\tau_{(\sigma x)_{\leq t+1}}(v))| \\
&\leq \frac{1}{\gamma} \sup_{x \in \sinfty} \sum_{t=1}^\infty \gamma^{t} |\beta(\tau_{x_{\leq t}}(v))| \\
&= \frac{1}{\gamma} \left(|\beta(v)| + \sup_{x \in \sinfty} \sum_{t=1}^\infty \gamma^{t} |\beta(\tau_{x_{\leq t}}(v))|\right) \\
&= \frac{1}{\gamma} s(v) = 0 \enspace.
\end{align*}
This implies $\tau_{\sigma}(v) \in \ker(s)$ for all $v \in \ker(s)$ and $\sigma \in \Sigma$. Therefore $\ker(s)$ is a bisimulation for $A$.

Now let $v \in W_A$. Since $W_A$ is contained in the kernel of $\beta$ and is invariant for all $\tau_\sigma$, we see that $\beta(\tau_x(v)) = 0$ for all $x \in \sstar$. Therefore, using the expression for $s$ given in Theorem~\ref{thm:closedform} we obtain $s(v) = 0$. This concludes the proof.
\end{proof}

\section{Proofs from Section~\ref{sec:paramcont}}\label{sec:appparamcont}

We first state an elementary lemma that we need in order to prove an upper
bound on $d_{\gamma}$.  This also played an important role in the
application of our bismulation pseudometric to spectral learning presented
in Section~\ref{sec:learning}.

\begin{lemma}\label{lem:recurrence}
Let $(s_l)_{l \in \N}$ be a sequence such that there exists a constant $a$
and a sequence $(b_l)_{l \in \N}$ satisfying $s_{l+1} \leq a s_{l} + b_{l}$
for all $l \geq 0$. Then for all $l \geq 0$ we have $s_{l+1} \leq a^{l+1}
s_0 + \sum_{i=0}^{l} a^{l-i} b_i$. 
\end{lemma}
\begin{proof}
Simple proof by induction on $l$.
\end{proof}
\begin{lemma}\label{lem:boundd}
Let $A = \wa$ and $A' = \wap$ be two weighted automata over the same
alphabet $\Sigma$ and the same vector space $V$. Let $M = \{\tau_\sigma\}
\cup \{\tau_\sigma'\}$ and $\rho = \rho(M)$. Suppose $\gamma < 1/\rho$ and
$\norm{\cdot}$ is a norm on $V$ such that for all $\sigma \in \Sigma$ we
have $\norm{\tau_\sigma}, \norm{\tau_\sigma'} \leq \theta$ for some
$\theta$ such that $\nu = \gamma \theta < 1$. Then we have the following: 
\begin{equation}\label{eqn:dAA}
d_\gamma(A,A') \leq  \frac{\norm{\alpha} \norm{\beta - \beta'}_*  +
  \norm{\beta'}_* \norm{\alpha - \alpha'}}{1 - \nu} + \frac{\gamma
  \norm{\alpha} \norm{\beta'}_* \max_{\sigma}\norm{\tau_\sigma -
    \tau_\sigma'}}{(1 - \nu)^2} 
\enspace.
\end{equation}
\end{lemma}

\begin{proof}

Fix $x \in \Sigma^\infty$ and given $l \geq 0$ define $D_l(x) = \sum_{t=0}^l\gamma^t |f_A(x_{\leq t}) - f_{A'}(x_{\leq t})|$.
By applying the triangle and H\"{o}lder inequalities to any term in the summation $D_l(x)$ we get 
\begin{equation}\label{eqn:absf}
|f_A(x_{\leq t}) - f_{A'}(x_{\leq t})| \leq \norm{\beta-\beta'}_*\norm{\tau_{x_{\leq t}}(\alpha)}
+\norm{\beta'}_*\norm{\tau_{x_{\leq t}}(\alpha)-\tau_{x_{\leq t}}'(\alpha')}
\enspace.
\end{equation}
Using the assumption on $\norm{\cdot}$ we can see that $\norm{\tau_{x_{\leq t}}(\alpha)}\leq \theta^t \norm{\alpha}$ for any $t \geq 0$. Now let $\varepsilon_\beta = \norm{\beta - \beta'}_*$ and $\varepsilon_t = \norm{\tau_{x_{\leq t}}(\alpha)-\tau_{x_{\leq t}}'(\alpha')}$. Plugging these definitions and the bound \eqref{eqn:absf} in $D_l$ we get
\begin{equation}
\label{eqn:Dl}
D_l(x) \leq \varepsilon_\beta \left(\sum_{t=0}^l \gamma^t \theta^t \right)+ \norm{\beta'}_* \left(\sum_{t=0}^l \gamma^t \varepsilon_t \right)
\enspace.
\end{equation}
Now we shall bound the term $s_l = \sum_{t=0}^l \gamma^t \varepsilon_t$. Suppose $x_{\leq t+1} = y \sigma$, where $y \in \Sigma^t$ and $\sigma \in \Sigma$. Let $\varepsilon_\tau = \max_{\sigma} \norm{\tau_\sigma - \tau_\sigma'}$. Using the triangle inequality we can show the following:
\begin{align*}
\varepsilon_{t+1}
&=
\norm{\tau_{y \sigma}(\alpha)-\tau_{y \sigma}'(\alpha')} \\
&=
\norm{\tau_{\sigma}(\tau_y(\alpha))- \tau_\sigma'(\tau_y'(\alpha'))} \\
&\leq 
\norm{\tau_\sigma(\tau_{y}'(\alpha')) - \tau_\sigma'(\tau_{y}'(\alpha'))} +
\norm{\tau_\sigma(\tau_{y}(\alpha)-\tau_{y}'(\alpha'))} \\
&\leq
\norm{\tau_\sigma - \tau_\sigma'} \norm{\tau_{y}'(\alpha')} +
\norm{\tau_\sigma} \norm{\tau_{y}(\alpha)-\tau_{y}'(\alpha')} \\
&\leq
\varepsilon_\tau \theta^t \norm{\alpha} + \theta \varepsilon_t
\enspace.
\end{align*}
We will now use the inequality above to show that $s_l$ satisfies a recurrence of the form considered in Lemma~\ref{lem:recurrence} for all $l \geq 0$:
\begin{align*}
s_{l+1} &= \varepsilon_0 +
\sum_{t=1}^{l+1} \gamma^t \varepsilon_t \\
&=
\varepsilon_0 +
\gamma \sum_{t=0}^l \gamma^t \varepsilon_{t+1} \\
&\leq
\varepsilon_0 +
\gamma \sum_{t=0}^l \gamma^t \left(\varepsilon_\tau \theta^t \norm{\alpha} + \theta \varepsilon_t \right) \\
&=
\gamma \theta s_l +
\varepsilon_0 + \gamma \varepsilon_\tau \norm{\alpha} \sum_{t=0}^l (\gamma \theta)^t \enspace.
\end{align*}
Let $\varepsilon_\alpha = \norm{\alpha - \alpha'}$ and note that $s_0 = \varepsilon_0 = \varepsilon_\alpha$. Thus, applying Lemma~\ref{lem:recurrence} with $a = \gamma \theta$ and $b_l = \varepsilon_\alpha + \gamma \varepsilon_\tau \norm{\alpha} \sum_{t=0}^l (\gamma \theta)^t$ to the sequence $s_l$ we get:
\begin{align*}
s_l &\leq
(\gamma \theta )^l \varepsilon_\alpha
+ \sum_{i=0}^{l-1} (\gamma \theta)^{l-1-i} \left(\varepsilon_\alpha + \gamma \varepsilon_\tau \norm{\alpha} \sum_{t=0}^i (\gamma \theta)^t \right) \\
&=
\varepsilon_\alpha \sum_{t=0}^l (\gamma \theta)^t
+ \gamma \varepsilon_\tau \norm{\alpha} \sum_{i=0}^{l-1} \left((\gamma \theta)^{l-1-i} \sum_{t=0}^i (\gamma \theta)^t\right) \\
&=
\varepsilon_\alpha \frac{1 - (\gamma \theta)^{l+1}}{1 - \gamma \theta} + \frac{\gamma \varepsilon_\tau \norm{\alpha}}{1 - \gamma \theta} \sum_{i=0}^{l-1} \left( (\gamma \theta )^{l-1-i} - (\gamma \theta )^l \right) \\
&=
\varepsilon_\alpha \frac{1 - (\gamma \theta)^{l+1}}{1 - \gamma \theta } + \frac{\gamma  \varepsilon_\tau \norm{\alpha}}{1 - \gamma \theta} \left( \frac{1 - (\gamma \theta )^{l}}{1 - \gamma \theta } - l (\gamma \theta )^l \right) \\
&=
\frac{\varepsilon_\alpha}{1 - \gamma \theta } + \frac{\gamma \varepsilon_\tau \norm{\alpha}}{(1 - \gamma \theta)^2} - (\gamma \theta)^l \left( \frac{\varepsilon_\alpha \gamma \theta  + l \gamma \varepsilon_\tau \norm{\alpha}}{1 - \gamma \theta} + \frac{\gamma \varepsilon_\tau \norm{\alpha} }{(1- \gamma \theta)^2} \right) \enspace.
\end{align*}
Plugging this bound into \eqref{eqn:Dl} and grouping the terms multiplied by $(\gamma \theta)^l$ into $R_l$ we get
\begin{equation}
D_l(x) \leq \frac{\varepsilon_\beta \norm{\alpha} + \varepsilon_\alpha \norm{\beta'}_*}{1 - \gamma \theta} 
+ \frac{\gamma \varepsilon_\tau \norm{\alpha} \norm{\beta'}_*}{(1 - \gamma \theta)^2} - (\gamma \theta)^l R_l \enspace.
\end{equation}
Finally, observing that $R_l = O(l)$ and using that $\gamma \theta = \nu < 1$, we take the limit $l \to \infty$ and obtain the desired bound using the closed form expression for $d_\gamma(A,A')$ given in Corollary~\ref{cor:closedformd}.
\end{proof}

Now we proceed to the proof of Theorem~\ref{thm:parcont}. The main
ingredient of this proof is the construction of a norm on $V$ satisfying
the conditions of Lemma~\ref{lem:boundd} uniformly for all $A_i$ with $i
\geq j_0$ for some $j_0 \in \N$. 

\begin{proof}[Proof of Theorem~\ref{thm:parcont}]
Let $A_i = \wan$ be a sequence of weighted automata converging to $A = \wa$
with respect to some norm $\norm{\cdot}$ on $V$ and suppose $\gamma <
1/\rho(A)$. For any $j \in \N$ we define the set 
\begin{equation*}
M_j = \{\tau_\sigma\}_{\sigma \in \Sigma} \cup \bigcup_{i \geq j}
\{\tau_{i,\sigma}\}_{\sigma \in \Sigma} \enspace. 
\end{equation*}
Since $\lim_{i \to \infty} \tau_{i,\sigma} = \tau_{\sigma}$ for all $\sigma
\in \Sigma$, the set $M_j$ is bounded for all $j \in \N$. Let $\rho_j =
\rho(M_j) = \rho(\bar{M}_j)$, where $\bar{M}_j$ is the compact set obtained
as the closure of $M_j$. Using the continuity of the joint spectral radius
on compact sets of operators we see that $\lim_{j \to \infty} \rho_j =
\rho(A)$. Thus, letting $\delta = 1 - \gamma \rho(A) > 0$, there exists a
constant $j_0 \in \N$ such that $|\rho_j - \rho(A)| < \delta /(4 \gamma)$
is satisfied for all $j \geq j_0$. Now we can apply
Theorem~\ref{thm:quasiextremal} to $\bar{M}_{j_0}$ with $\eta = \delta / (4
\gamma)$ to find a norm $\norm{\cdot}'$ on $V$ such that
$\norm{\tau_\sigma}' \leq \rho(A) + \delta / (2 \gamma)$ and
$\norm{\tau_{i,\sigma}}' \leq \rho(A) + \delta / (2 \gamma)$ for all
$\sigma \in \Sigma$ and all $i \geq j_0$. Taking $\theta = \rho(A) + \delta
/ (2 \gamma)$ we see that $\gamma \theta = \gamma \rho(A) + \delta / 2 <
\gamma \rho(A) + \delta = 1$. Hence, we are under the hypotheses of
Lemma~\ref{lem:boundd} and we have that the following holds for all $i \geq
j_0$:
\begin{equation}\label{eqn:dAAi}
d_\gamma(A,A_i) \leq  \frac{\norm{\alpha}' \norm{\beta - \beta_i}'_{*}  +
  \norm{\beta_i}'_{*} \norm{\alpha - \alpha_i}'}{1 - \nu} + \frac{\gamma
  \norm{\alpha}' \norm{\beta_i}'_{*} \max_{\sigma}\norm{\tau_\sigma -
    \tau_{i,\sigma}}'}{(1 - \nu)^2} 
\enspace,
\end{equation}
where $\nu = \gamma \theta = \gamma \rho(A) + \delta / 2$.

Now recall that all norms in a finite dimensional vector space are
equivalent. Therefore, we can find a pair constants $0 < c \leq C$ such
that $c \norm{v} \leq \norm{v}' \leq C \norm{v}$ holds for all $v \in V$
and $C^{-1} \norm{w}_* \leq \norm{w}'_* \leq c^{-1} \norm{w}_*$ for all $w
\in V^*$. Plugging these inequalities in \eqref{eqn:dAAi} we see that for
all $i \geq j_0$ we have 
\begin{equation*}
d_\gamma(A,A_i) \leq  \frac{C (\norm{\alpha} \norm{\beta - \beta_i}_{*}  +
  \norm{\beta_i}_{*} \norm{\alpha - \alpha_i})}{c (1 - \nu)} + \frac{C^2
  \gamma \norm{\alpha} \norm{\beta_i}_{*} \max_{\sigma} \norm{\tau_\sigma -
    \tau_{i,\sigma}}}{c (1 - \nu)^2} 
\enspace.
\end{equation*}
Since the sequence of automata $(A_i)$ converges to $A$ with respect to
$\norm{\cdot}$, we conclude that $\lim_{i \to \infty} d_{\gamma}(A,A_i) =
0$. 
\end{proof}

\section{Proofs from Section~\ref{sec:inputcont}}\label{sec:appinputcont}

\begin{proof}[Proof of Theorem ~\ref{thm:gcontinuous}]
Let $A$ be weighted automaton such that $\gamma < 1/\rho(A)$ and let $(A_i)_{i \in \N}$ be a sequence of weighted automata converging to $A$ with respect to $d_\gamma$. Note that for any $i \in \N$ we have the following:
\begin{align*}
\sup_{x \in \sstar} \frac{ | f_A(x) - f_{A_i}(x) |}{g(|x|)}
= \sup_{x \in \sstar} \frac{ | f_A(x) - f_{A_i}(x) | \gamma^{|x|}}{g(|x|)\gamma^{|x|}}
\leq \sup_{x \in \sstar} \frac{d_\gamma(A,A_i)}{g(|x|)\gamma^{|x|}}
= \sup_{l \in \N} \frac{d_\gamma(A,A_i)}{g(l) \gamma^{l}} \enspace.
\end{align*}
Now note that $g(l) > 0$ and $g(l) = \Omega(\gamma^{-l})$ implies $\inf_{l \in \N} g(l) \gamma^{l} > 0$. Using the assumption that $\lim_{i \to \infty} d_\gamma(A,A_i) = 0$ we now see that \eqref{eqn:inpcon} is satisfied.
\end{proof}

\begin{proof}[Proof of Theorem ~\ref{thm:notgcontinuous}]
Let $\Sigma = \{a\}$ be an alphabet with one symbol and let $A_i = \waai$ with $\tau_i = 1 + 2^{-i}$ and $\alpha = \beta = 1$ be the weighted automaton shown on the left of Figure~\ref{fig:WAinputcts}, and let $A = \waa$ with $\tau = 1$ be the weighted automaton shown on the right of Figure~\ref{fig:WAinputcts}. For any $i > \log_2(\gamma/(1-\gamma))$ we have $\gamma \tau_i < 1$. Hence, we can write
\begin{align*}
d_\gamma(A,A_i)
=\sup_{x \in \sinfty} \sum_{t \geq 0} \gamma^t | \tau^t - \tau_{i}^t |
=\sum_{t \geq 0} \gamma^t \left( (1+2^{-i})^{t} - 1 \right)
= \frac{1}{1-\gamma (1+2^{-i})} - \frac{1}{1- \gamma}
\enspace.
\end{align*}
Therefore we see that $\lim_{i \to \infty} d_\gamma(A,A_i) = 0$. Now let us show that for these automata the limit in \eqref{eqn:inpcon} is not zero for any $g(l) = c^{o(l)}$ with $c > 1$. Indeed, we can write
\begin{align*}
\sup_{x \in \sstar} \frac{|f_A(x) - f_{A_i}(x)|}{g(|x|)} &=
\sup_{x \in \sstar} \frac{(1 + 2^{-i})^{|x|} - 1}{c^{o(|x|)}} =
\sup_{l \in \N} \frac{(1 + 2^{-i})^{l} - 1}{c^{o(l)}} \\
&\geq
\sup_{l \in \N} \frac{(1 + 2^{-i})^{l}}{c^{o(l)}} - \sup_{l \in \N} \frac{1}{c^{o(l)}} = \infty \enspace,
\end{align*}
where the last equality uses that $\frac{(1 + 2^{-i})^{l}}{c^{o(l)}} = \omega(1)$ and $\frac{1}{c^{o(l)}} = O(1)$ with respect to $l \to \infty$. Therefore $d_\gamma$ is not input $g$-continuous for these choices of $g$.
\end{proof}

\begin{figure}
\caption{Two weighted automata with $\Sigma=\{a\}$ and initial weight $\alpha=1$.}
\label{fig:WAinputcts}
\begin{center}
\begin{tikzpicture}
[->,>=stealth',shorten >=1pt,auto,node distance=4.3cm,
                    semithick]
  \tikzstyle{every state}=[]

  \node[state, circle split]	(q1) 	
  {$q_1$ \nodepart{lower} $1$};
  \node[state, circle split]	(q2) [right of=q1] 
  {$q_2$ \nodepart{lower} $1$};

  \path (q1) edge [loop right]          node {a:$\tau_i$} (q1)
        (q2) edge [loop right]          node {a:$1$}      (q2)
        ;
\end{tikzpicture}
\end{center}
\end{figure}
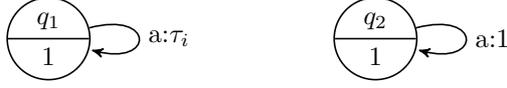

\section{Proofs from Section~\ref{sec:learning}}\label{sec:applearning}

\begin{proof}[Proof of Theorem~\ref{thm:spectral}]
Let $M = \{\tau_\sigma\}_{\sigma \in \Sigma}$ and let $\norm{\cdot}$ be a norm on $V$ obtained from Theorem~\ref{thm:quasiextremal} with $M$ and a small enough constant $\eta > 0$. Let $\hat{M} = \{\tau_\sigma\}_{\sigma \in \Sigma} \cup \{\hat{\tau}_\sigma\}_{\sigma \in \Sigma}$. Let $d_H$ denote the Hausdorff distance between sets of linear operators induced by $\norm{\cdot}$. Since $M$ is irreducible we can use the local Lipschitz continuity of the joint spectral radius to see that there exists a constant $c_M > 0$ depending only on $M$ such that the following holds:
\begin{align*}
|\rho(M) - \rho(\hat{M})| &\leq c_M d_H(M,\hat{M}) = c_M \max\left\{\sup_{\tau \in M} \inf_{\tau' \in \hat{M}} \norm{\tau - \tau'}, \sup_{\tau' \in \hat{M}} \inf_{\tau \in M} \norm{\tau - \tau'} \right\} \\
&\leq c_M \max_{\sigma \in \Sigma} \norm{\tau_\sigma - \hat{\tau}_\sigma} \enspace.
\end{align*}
Note that by Lemma~\ref{lem:learn} we have $\max_{\sigma \in \Sigma} \norm{\tau_\sigma - \hat{\tau}_\sigma} = O(\norm{H - \hat{H}}_2)$. Thus, by making $\norm{H - \hat{H}}_2$ small enough we can assume that $\gamma \rho(\hat{M}) < 1$. Using this fact and our choice of $\eta$ we can apply Lemma~\ref{lem:boundd} to see that $d_{\gamma}(A,\hat{A}) \leq O(\norm{H - \hat{H}}_2)$. Furthermore, the hidden constants in the big-$O$ notation depend on: the Hankel block indices $B = (P,S)$ and the size of the alphabet $|\Sigma|$ through Lemma~\ref{lem:learn}; on $A$ through Lemma~\ref{lem:learn}, the norm $\norm{\cdot}$, and the constant $c_M$; and on $\gamma$ through Lemma~\ref{lem:boundd}.
\end{proof}

\end{document}